\documentclass[conference,a4paper,10pt]{IEEEtran}
\usepackage[T1]{fontenc} 
\pdfoutput=1

\IEEEoverridecommandlockouts

\usepackage{cite}
\usepackage{amsmath,amssymb,amsfonts}
\usepackage[ruled,vlined,linesnumbered]{algorithm2e}
\usepackage{algpseudocode}
\usepackage{etoolbox}  
\usepackage{bbm}
\usepackage{amsmath}
\usepackage{tabularx}
\usepackage{amssymb}
\usepackage{threeparttable}
\usepackage{booktabs}
\usepackage{amsthm}
\usepackage{mathtools}
\usepackage{bm}

\usepackage{tikz}
    \usetikzlibrary{shapes.arrows}

\usepackage{pgfplots}
\pgfplotsset{
compat=1.3,
legend style={font=\footnotesize, fill opacity=0.7,  draw opacity=1, text opacity=1, draw=white!15!black, legend cell align=left, align=left}, 
width=6cm, 
height=6cm,
yminorticks=false,
xminorticks=false,
title style={font=\small},
tick style={color=black},
tick label style={font=\small},
grid style={line width=.1pt, draw=gray!20},
major grid style={line width=.1pt,draw=gray!20},
}
\usepgfplotslibrary{colorbrewer}
\usepgfplotslibrary{groupplots}

\usepackage{adjustbox}
\usepackage{url}
\usepackage{gincltex}
\usepgfplotslibrary{fillbetween}
\usetikzlibrary{plotmarks}
\usetikzlibrary{patterns}
\usepackage[font=footnotesize,caption=false]{subfig} 
\usepackage{graphicx}
\usepackage{textcomp}
\usepackage{xcolor}
\def\BibTeX{{\rm B\kern-.05em{\sc i\kern-.025em b}\kern-.08em
		T\kern-.1667em\lower.7ex\hbox{E}\kern-.125emX}}

\pgfplotsset{compat=1.15}
		
\usepackage{glossaries}

\newacronym{af}{AF}{array factor}
\newacronym{urllc}{URLLC}{Ultra-Reliable Low-Latency Communications}
\newacronym{qos}{QoS}{Quality of Service}
\newacronym{pdf}{PDF}{Probability Density Function}
\newacronym{cdf}{CDF}{Cumulative Density Function}
\newacronym{iot}{IoT}{Internet of Things}
\newacronym{los}{LoS}{Line of Sight}
\newacronym{pec}{PEC}{Packet Erasure Channel}
\newacronym[plural=RATs,firstplural=Radio Access Technologies (RATs)]{rat}{RAT}{Radio Access Technology}
\newacronym{simo}{SIMO}{Single-Input Multiple-Output}
\newacronym{bs}{BS}{Base Station}
\newacronym{ris}{RIS}{Reconfigurable Intelligent Surface}
\newacronym{risc}{RISC}{\gls{ris} Controller}
\newacronym{snr}{SNR}{Signal to Noise Ratio}
\newacronym{upa}{UPA}{Uniform Planar Array}
\newacronym{ula}{ULA}{Uniform Linear Array}
\newacronym{ue}{UE}{User Equipment}
\newacronym{csi}{CSI}{Channel State Information}
\newacronym{dl}{DL}{Downlink}

\definecolor{violet}{rgb}{0.6,0,0.6}%
\definecolor{orange_D}{rgb}{1,0.3,0}%
\definecolor{cyan}{rgb}{0,0.67,0.64}%
\definecolor{red}{rgb}{0.9,0,0}%
\definecolor{green_D}{rgb}{0,0.5,0}%
\definecolor{yellow}{rgb}{1,0.8,0}
\definecolor{amaranth}{rgb}{0.9, 0.17, 0.31}

\newcommand{\T}{^{\intercal}}     
\newcommand{\mc}[1]{\mathcal{#1}}   
\newcommand{\mb}[1]{\mathbf{#1}}    

\newtheorem{theorem}{Theorem} 
\newtheorem{proposition}{Proposition}

\definecolor{color1}{HTML}{C9A847}
\definecolor{color2}{HTML}{D3522C}
\definecolor{color3}{HTML}{620000}

\def \mwidth{0.2\linewidth}

\begin{document}

\title{Efficient URLLC with a Reconfigurable Intelligent Surface and Imperfect Device Tracking}


\author{\IEEEauthorblockN{Fabio Saggese\IEEEauthorrefmark{1}, Federico Chiariotti\IEEEauthorrefmark{1}\IEEEauthorrefmark{2}, Kimmo Kansanen\IEEEauthorrefmark{1}\IEEEauthorrefmark{3}, and Petar Popovski\IEEEauthorrefmark{1}}

\IEEEauthorblockA{\IEEEauthorrefmark{1}Department of Electronic Systems, Aalborg University, Denmark (\{fasa, fchi, kimkan, petarp\}@es.aau.dk)}
\IEEEauthorblockA{\IEEEauthorrefmark{2} Department of Information Engineering, University of Padova, Italy}
\IEEEauthorblockA{\IEEEauthorrefmark{3}Department of Electronic Systems, Norwegian University of Science and Technology}\thanks{This work was partly supported by the Villum Investigator grant ``WATER'' from the Villum Foundation, Denmark, and by the Horizon 2020 ``RISE-6G'' project, financed by the European Commission under grant no. 101017011.}}

\maketitle

\begin{abstract}
The use of \gls{ris} technology to extend coverage and allow for better control of the wireless environment has been proposed in several use cases, including \gls{urllc} communications. However, the extremely challenging latency constraint makes explicit channel estimation difficult, so positioning information is often used to configure the \gls{ris} and illuminate the receiver device. In this work, we analyze the effect of imperfections in the positioning information on the reliability, deriving an upper bound to the outage probability. We then use this bound to perform power control, efficiently finding the minimum power that respects the \gls{urllc} constraints under positioning uncertainty. The optimization is conservative, so that all points respect the \gls{urllc} constraints, and the bound is relatively tight, with an optimality gap between 1.5 and 4.5~dB.
\end{abstract}

\IEEEpeerreviewmaketitle
\glsresetall

\section{Introduction}\label{sec:intro}
 
A growing body of research has recognized the potential of \glspl{ris} to provide improvements in wireless communications by imposing low-power real-time control on the propagated wireless signals~\cite{Wu2021ristutorial}. As such, it can result in various performance benefits, such as improved coverage, increased data rate, and mitigation of multi-user interference~\cite{Wu2021ristutorial, bjornson2021signalprocessing}. 

Intuitively, a \gls{ris} is a good match for \gls{urllc}, since it can act as a full-duplex relay. It can thus support two-way exchanges without incurring additional latency, as there is no need to change the configuration when the communication direction is changed. Nevertheless, the \gls{ris} configuration needs to be optimized in order to create a favorable wireless channel between the nodes of interest~\cite{bjornson2021signalprocessing}. 
The potential of \gls{ris} in the context of \gls{urllc} has been analyzed in~\cite{Hashemi2021risurllc}, where the authors prove that the reliability of the system is improved by adding an \gls{ris}.
Other works on the topic focus on the optimization of the system: in~\cite{Ghanem2021misoris}, the authors design a low-complexity joint beamforming and phase shift optimization algorithm; in~\cite{Melgarejo2020grantfreeURLLC}, the authors propose a grant-free uplink access paradigm based on a mix of resource allocation schemes and receiver design. However, the above literature is conditioned on the knowledge of the \gls{csi}. In real systems, \gls{ris} channel estimation procedures have a complexity proportional to the number of \gls{ris} elements~\cite{Li2021ce, mengnan2022survey}, and are not practical for \gls{urllc} traffic, as the procedure can add a significant delay. 

A possible alternative, especially in cases where a strong \gls{los} component dominates the propagation, is to use the positioning information of the terminals to configure the \gls{ris} beamformer to illuminate the receiver~\cite{Jamali2022lowoverhead}. However, available position information is subject to uncertainty, and, hence, robust optimizations accounting for such error are needed. A formal analysis on the statistical relation between reliability and uncertainty for a generic communication system is carried out in~\cite{Kallehauge2022tradeoff}. A design of \gls{ris} phase shift optimization taking into account such uncertainty to maximize the average spectral efficiency with no reliability constraint (which makes it unsuitable for \gls{urllc}) is presented in~\cite{Abrardo2021positioning}.

\begin{figure}[t]
    \centering
    \includegraphics[width=7cm]{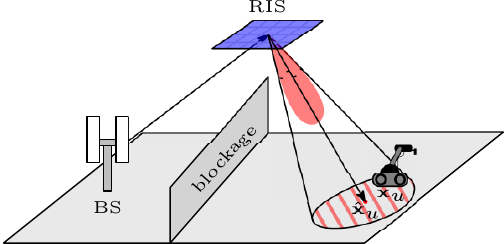}
    \caption{Scenario of interest.}\vspace{-0.4cm}
    \label{fig:scenario}
\end{figure}

This paper addresses the problem of \gls{urllc} communication when the \gls{ris} is configured to beam energy towards a \gls{ue} based on the tracking of its position at the \gls{bs}. The setup is depicted in Fig.~\ref{fig:scenario}, in which the \gls{bs} transmits in the downlink. For the depicted \gls{ue}, the dominant portion of the wireless signal arrives through a \emph{controllable} path, which the \gls{ris} can affect by changing its configuration.
The uncertainty of the \gls{ue}'s position is due to the noisy tracking process, which is often based on Kalman or particle filters. We propose a method that embeds the position uncertainty into the overall \gls{urllc} reliability requirement, determines the \gls{ris} configuration and sets the \gls{bs} transmit power to meet the reliability constraint.\footnote{The simulation code for the paper is available at \url{https://github.com/AAU-CNT/efficient-ris-aided-urllc}}

The rest of the paper is divided as follows: first, Sec.~\ref{sec:sys} presents the system model. The reliability bound and power optimization are derived in Sec.~\ref{sec:power-control}, and numerical performance results are described in Sec.~\ref{sec:results}. Finally, Sec.~\ref{sec:conc} concludes the paper and presents some possible avenues of future work.

\paragraph*{Notation} 
$\text{mod}(a, M)$ is the $a$ modulo $M$ operation; $\lfloor a \rfloor$ represent the nearest lower integer of $a$.
Lower and upper case boldface letters denote vectors $\mathbf{x}$ and matrices $\mathbf{X}$, respectively; the Euclidean norm of $\mathbf{x}$ is $\lVert\mathbf{x}\rVert$. $\mc{P}(e)$ is the probability that event $e$ occurs; $\mc{CN}(\bm{\mu},\mb{R})$ is the complex Gaussian distribution with mean $\bm{\mu}$ and covariance matrix $\mb{R}$; $\mc{R}(K, \Omega)$ is the Rice distribution with shape parameter $K$ and scale parameter $\Omega$; $\mathbb{E}[\cdot]$ is the expected value operation.

\section{System Model}\label{sec:sys}
We consider an industrial \gls{dl} \gls{urllc} communication scenario in which a single-antenna \gls{bs} has to communicate to a single-antenna mobile \gls{ue}. To ensure full coverage, a $N$-element \gls{ris} is deployed on the ceiling of the factory.
We define the three-dimensional coordinate system $O_1$, whose origin lies in the center of the \gls{ris} $\mathbf{x}_{r}=(0,0,0)\T$. The $x$ and $y$ axes are parallel to the horizontal dimensions of the \gls{ris}, and the $z$ axis points towards the floor of the factory. A depiction of coordinate system $O_1$ is given in Fig.~\ref{fig:O1}. 

It is assumed that there is a known one-to-one mapping between the locations of the UE and the \gls{ris} configurations. For the propagation scenario in this paper, the mapping is explicitly available as analytical expressions. In the considered industrial scenario with more complicated propagation characteristics, the availability of the mapping implies that a suitable calibration process has taken place, which can map the regions of the floor in which the \gls{ris} provides a significant benefit, as the direct path between \gls{bs} and \gls{ue} is blocked or weak. 

\subsection{Communication System Model}
The available space for the \gls{ue} is delimited by a square floor of area $D^2$~m$^2$ and the ceiling is at a height of $h$~m. The position  of the \gls{bs} is $\mb{x}_b = (x_b, y_b, z_b)\T$, and a \gls{los} path with the \gls{ris} exists. The \gls{ue} moves around the factory floor, so that its coordinates at time $t$ are $\mathbf{x}_u(t)=(x_{u}(t),y_{u}(t),h)\T$. 
Assuming a square \gls{ris} for simplicity of presentation, each element has position on $O_1$ given by 
$\mb{r}_n = d (\text{mod}(n -1, \sqrt{N}) - \frac{\sqrt{N}-1}{2}, \lfloor \frac{ n -1 }{\sqrt{N}} \rfloor - \frac{\sqrt{N}-1}{2}, 0)\T$
where $d < \lambda$ is the spacing between the center of neighboring elements and $\lambda$ is the carrier wavelength.  We assume $h \ge \frac{2}{\lambda} d^2 N^2$ in order to assure far-field propagation regardless the position of the \gls{ue}.
Each \gls{ris} element influences the incoming signal by inducing a phase shift $\phi_n$, and we assume that the attenuation imposed by each element of the \gls{ris} is strictly equal to 1. The vector containing the phase shifts of each \gls{ris} element is denoted as $\bm{\phi} = [e^{j\phi_1},\dots,e^{j\phi_N}]\T$, and it is referred as \emph{phase profile} or \emph{configuration} in the remainder of the paper.
The \gls{ris} phase profile is controlled by the \gls{bs}, which receives information about \gls{ue}s positioning through an out-of-band control channel. The aim of the \gls{bs} is then to optimize its own transmitting power and the \gls{ris} phase profile to communicate efficiently to the \gls{ue} while respecting the \gls{urllc} constraints.

We assume that the \gls{bs} tracks the \gls{ue}'s movements through a Kalman-like filter~\cite{kalman1960new}, which is a common assumption in indoor and outdoor tracking applications~\cite{vatansever2017visible}, so that the \gls{pdf} of the estimated position of the \gls{ue} at time $t$ is a bivariate Gaussian random variable:
\begin{equation}
    p_{\mathbf{x}_u(t)}(\mathbf{x})=\frac{1}{2\pi|\bm{\Sigma}(t)|}e^{-\frac{1}{2}(\mathbf{x}-\hat{\mathbf{x}}_u(t))\T\bm{\Sigma}(t)(\mathbf{x}-\hat{\mathbf{x}}_u(t))},
\end{equation}
where $\hat{\mathbf{x}}_u(t)$ is the estimated position at time $t$ and $\bm{\Sigma}(t)$ is the covariance matrix of the Kalman filter.

\begin{figure}
    \centering
    \includegraphics[width=4.3cm]{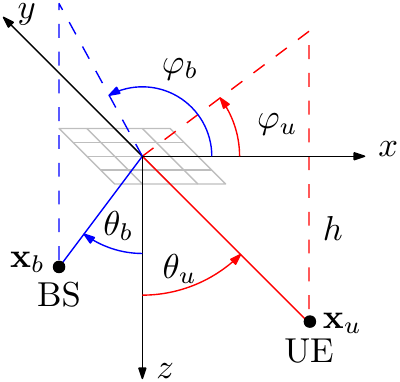}
    \caption{Coordinate system $O_1$ of the scenario.}\vspace{-0.4cm}
    \label{fig:O1}
\end{figure}

We can convert the \gls{ue}'s and \gls{bs}' positions from Cartesian to spherical coordinates $\mathbf{z}_u(t)=(r_u(t), \theta_u(t), \varphi_u(t))\T$, and $\mathbf{z}_b(t)=(r_b(t), \theta_b(t), \varphi_b(t))\T$ using the \gls{ris}'s center as the center of the sphere. The three coordinates represent the radius $r_i\in[0, \sqrt{D^2/2 + h^2}]$ from the center to the \gls{ris} to the user (\gls{bs}), the elevation angle $\theta_i\in[-\pi/2, \pi/2 ]$ computed from the $z$ axis to the \gls{ue}'s (\gls{bs}') position, and the  azimuth angle $\varphi_i\in[0, 2\pi]$ identifying the point on the $x-y$ plane, respectively, with $i\in\{b,u\}$. In the following, we omit the time index $t$ for brevity.

Assuming a transmission bandwidth well below the channel coherence bandwidth, each transmitted symbol will experience a channel which depends on the position of the \gls{ue} and the \gls{bs}, the \gls{ris} configuration loaded in that moment, and the frequency flat short term fading. Without loss of generality, we assume that the \gls{bs} transmits a single symbol $x$ with power $P$ towards the \gls{ue}; the received signal can be modeled as
\begin{equation} \label{eq:signal}
    y = \sqrt{\beta(\mathbf{x}_u) P} g_{b} g_{u} \mb{a}_b\T \mathrm{diag}(\bm{\phi}) \mb{a}_u x + n
\end{equation}
where $|g_{b}|, |g_{u}| \sim \mc{R}(K, 1)$ are the short term fading realizations for the \gls{bs}-\gls{ris} and \gls{ris}-\gls{ue} paths, and $n\sim \mc{CN}(0, \sigma^2)$ is the receiver noise. The path loss term is given by 
\begin{equation}
    \beta(\mathbf{x}_u) = \beta_0^2 G_b G_u \left(\frac{d_0^2}{\lVert\mb{x}_u\rVert \, \lVert\mb{x}_b\rVert}\right)^\xi 
\end{equation}
where $\beta_0$ is the path loss at a reference distance $d_0$, $G_b$ and $G_u$ are the antenna gain of the \gls{bs} and \gls{ue}, respectively; $\xi$ is the path loss exponent. Note that $\beta_0$ and $d_0$ are squared due to the double path \gls{bs}-\gls{ris} and \gls{ris}-\gls{ue}.
The steering vectors $\mb{a}_b$ and $\mb{a}_u$ represent the angle of arrival from the \gls{bs} to each \gls{ris} element and the angle of departure from each element to the \gls{ue}, respectively. The steering vectors are:
\begin{equation}
    [\mb{a}_i]_n = e^{j \frac{2 \pi}{\lambda} \frac{\mb{x}_i\T \mb{r}_n}{\lVert\mb{x}_i \rVert}} = e^{j \frac{2 \pi}{\lambda} \left(r_{n,1} \sin\theta_i\cos\varphi_i + r_{n,2} \sin\theta_i\sin\varphi_i\right)},   
\end{equation}
with  $i \in\{b, u\}$. The \gls{snr} is then:
\begin{equation} \label{eq:snr}
    \gamma = \frac{P}{\sigma^2} \beta |g_b g_u|^2 N^2 |A(\boldsymbol{\phi}|\hat{\mathbf{x}}_u,\mathbf{x}_u)|^2,
\end{equation}
where $A(\bm{\phi}|\hat{\mathbf{x}}_u,\mathbf{x}_u) = \frac{1}{N} \mb{a}_b\T \mathrm{diag}(\boldsymbol{\phi}) \mb{a}_u$ denotes the \gls{ris} \gls{af}, while $|A(\boldsymbol{\phi}|\hat{\mathbf{x}}_u,\mathbf{x}_u)|^2$ denotes the \gls{af} gain.

\subsection{Array Factor: Pointing and Beamwidth}
We define two indexes spanning through the \gls{ris} elements in the horizontal ($x$) and vertical ($y$) dimensions as $\ell = \text{mod}(n - 1, \sqrt{N})$ and $k=\big\lfloor \frac{ n -1 }{\sqrt{N}} \big\rfloor$, respectively, and rewrite the \gls{af} as
\begin{equation} \label{eq:af1}
\begin{aligned}
    A(\boldsymbol{\phi}|\hat{\mathbf{x}}_u,\mathbf{x}_u) &= \frac{1}{N} \sum_{\ell=0}^{\sqrt{N}-1} \sum_{k=0}^{\sqrt{N}-1} e^{j\phi_{\ell, k}} \\
    &\cdot e^{j \frac{2 \pi}{\lambda} d\left(\ell - \frac{\sqrt{N}-1}{2} \right) \left(\sin\theta_u\cos\varphi_u + \sin\theta_b\cos\varphi_b\right)}  \\
    & \cdot e^{j\frac{2\pi}{\lambda} d\left(k - \frac{\sqrt{N}-1}{2}\right) \left(\sin\theta_u\sin\varphi_u + \sin\theta_b \sin\varphi_b \right)},
    \end{aligned}
\end{equation}
where $\phi_{\ell,k} = \phi_n$ using the appropriate index.
Without loss of generality, the phase shift impressed by each element can be expressed as $\phi_{\ell,k} = \ell \phi_{x} + k \phi_{y}$. In this way, we can compensate for the (known) position of the \gls{bs}, while pointing toward the direction given by $\hat{\theta}$ and $\hat{\phi}$, by setting
\begin{equation}\label{eq:phaseprofile1}
\begin{aligned}
    \phi_x = - \frac{\pi d}{\lambda} \left(\sin\hat{\theta}\cos\hat{\varphi} + \sin\theta_b\cos\varphi_b\right), \\
    \phi_y = - \frac{\pi d}{\lambda} \left(\sin\hat{\theta}\sin\hat{\varphi} + \sin\theta_b\sin\varphi_b\right). \\
\end{aligned}
\end{equation}
Inserting~\eqref{eq:phaseprofile1} into~\eqref{eq:af1}, the \gls{af} can be rewritten as~\cite{balanis2015antenna}
\begin{equation} \label{eq:af}
\begin{aligned}
    A(\bm{\phi}|\hat{\mathbf{x}}_u,\mathbf{x}_u) = \frac{e^{j \frac{\sqrt{N}-1}{2} (\phi_x + \phi_y)}}{N}   \frac{\sin(\sqrt{N} f_x)}{\sin(f_x)} \frac{\sin(\sqrt{N} f_y)}{\sin(f_y)},
    \end{aligned}
\end{equation}
with $f_x = \frac{\pi d}{\lambda} (\sin\theta_u\cos\varphi_u - \sin\hat{\theta}\cos\hat{\varphi})$ and $f_y =\frac{\pi d}{\lambda} (\sin\theta_u\sin\varphi_u - \sin\hat{\theta}\sin\hat{\varphi})$.
If perfect knowledge of the \gls{ue}'s position is available, a trivial solution to maximize the \gls{af} is to set $\hat{\theta} = \theta_u$, $\hat{\varphi} = \varphi_u$; however, a positioning error is always present. Therefore, we resort to evaluate the illuminated region (on the floor) $\mc{G}(A_0)$ in which the \gls{af} gain is at least equal to $A_0$ after setting $\hat{\theta} = \hat{\theta}_u$, $\hat{\varphi} = \hat{\varphi}_u$.

The \gls{af} of a square \gls{upa} when pointing toward the direction given by $\hat{\theta}$ and $\hat{\varphi}$, i.e., $\hat{\mb{w}} = (\sin\hat{\theta}\cos\hat{\varphi}, \sin\hat{\theta}\sin\hat{\varphi},\cos\hat{\theta})\T$ is given by~\eqref{eq:af}. The 3D beamwidth generating an \gls{af} gain of $A_0\in(0,1]$ can then be approximated by the angles~\cite[Section 6.10]{balanis2015antenna}
\begin{equation} \label{eq:3Dbeamwidth}
    \begin{aligned}
        \Delta \theta(A_0) = \frac{\Delta\Theta(A_0)}{\cos\hat{\theta}}, \quad \Delta \varphi(A_0) =  \Delta\Theta(A_0).
    \end{aligned}
\end{equation}
Thus, the points in the 3D space with an \gls{af} gain of $A_0$ lie on the surface of an elliptic cone whose vertical axis is the pointing direction $\hat{\mb{w}}$, the major diameter ($2a$) is generated by the angle $\Delta \theta(A_0)$ in the elevation plane is defined by $\varphi = \hat{\varphi}$, and the minor diameter ($2b$) is generated by the angle $\Delta \theta(A_0)$ on the plane perpendicular to the elevation one (see~\cite[Fig. 6.38]{balanis2015antenna}).
In~\eqref{eq:3Dbeamwidth}, $\Delta\Theta(A_0)$ is the beamwidth of the \gls{ula} spanning through the $x$ (or $y$) dimension providing an \gls{af} gain of $A_0$, whose approximation is given
in the following proposition.
\begin{proposition}
The beamwidth of a \gls{ula} given an \gls{af} target gain $A_0$ can be approximated as
\begin{equation}
    \Delta\Theta(A_0) \approx \arcsin\left(2 \frac{\lambda x(A_0)}{\pi d \sqrt{N}}\right),
\end{equation}
where $x(A_0) = \{x \,|\, \text{sinc}(x) = A_0\}$.
\end{proposition}
\begin{proof}
In the main lobe, the \gls{ula} \gls{af} is $\approx \text{sinc}(\sqrt{N} f_x)$~\cite{balanis2015antenna}, where $\varphi = \hat{\varphi} = 0$ because the elements span the $x$ axis. Neglecting the pointing effect, i.e.,  $\hat{\theta} = \pi / 2$, and solving $ \sqrt{N} f_x = x(A_0)$ with respect to $\theta$ completes the proof.
\end{proof}
The arguments $x(A_0)$ can be easily obtained by numerical simulations for the desired values of $A_0$.

\section{Power Control for URLLC}
\label{sec:power-control}
While most \gls{urllc} packets are short, the finite blocklength effects disappear when only statistical knowledge of the \gls{csi} is available, as proven in~\cite{Durisi2016}. Hence, we can use Shannon's capacity formula to derive the minimum required \gls{snr} to reliably deliver the data, and perform power control for the \gls{urllc} transmission. 
If we consider a packet of length $L$ bits, which has to be transmitted with a maximum latency $T$ over bandwidth $B$, the minimum required \gls{snr} $\gamma_0$ is
\begin{equation}
    \gamma_0=2^{\frac{L}{BT}}-1.\label{eq:min_snr}
\end{equation}
The actual \gls{snr} in~\eqref{eq:snr} includes two independent random components: the first is the fading, which is given by the product of the fading on the \gls{bs}-\gls{ris} and \gls{ris}-\gls{ue} channels, while the second depends on the actual position of the \gls{ue}. If the \gls{ue}'s coordinates are given by $\mathbf{x}_u$, the average \gls{snr} $\hat{\gamma}(\hat{\mathbf{x}}_u,\mathbf{x}_u)$ is given by:
\begin{equation}\label{eq:av_snr}
    \hat{\gamma}(\hat{\mathbf{x}}_u,\mathbf{x}_u)=\frac{P}{\sigma^2} \beta(\mathbf{x}_u) \mathbb{E}[|g_m g_b|^2] N^2 |A(\boldsymbol{\phi}|\hat{\mathbf{x}}_u,\mathbf{x}_u)|^2,
\end{equation}
and we know that $\mathbb{E}[|g_mg_b|^2]=1$. The two random terms (the fading and the \gls{ue}'s position) can be considered as independent. The \gls{bs} can then optimize the transmission power $P$ to minimize energy consumption while meeting the \gls{urllc} requirements.

The intuition for our procedure is the following: first, we find an upper bound to the outage probability by considering either a deep fading event or a large positioning error as a failure (without computing the intersection, or the possibility that a lucky fading gain might compensate for a larger positioning error). We then try to find the minimum gain that allows the resulting beam to illuminate the region in which the \gls{ue} might be, and invert the resulting values to find a transmission power that ensures the \gls{urllc} constraints are met.

\subsection{Reliability Bound}

The overall problem is complex, as the distribution of the instantaneous \gls{snr} is extremely difficult and must be obtained numerically, making the computation difficult: as \gls{urllc} requirements limit the computational effort that can be spent in optimizing the system before transmission, this makes a direct calculation infeasible. However, we can compute a lower bound to reliability by separating the two components:
\begin{theorem}\label{th:prob_bound}
Let $G_0>0$ be a positive value. If $\mc{P}(|g_ug_b|^2\leq G_0)=\delta$, we have:
\begin{equation}
    \mc{P}(\gamma<\gamma_0)\leq\delta+\mc{P}\left(\hat{\gamma}\leq\frac{\gamma_0}{G_0}\right).\label{eq:th1_cond}
\end{equation}
\end{theorem}
\begin{proof}
We can consider two cases:
\begin{enumerate}
    \item In the first case, $|g_ug_b|^2> G_0$. We can then express the following bound:
    \begin{equation}
        \gamma=\hat{\gamma}|g_ug_b|^2>\hat{\gamma}G_0.
    \end{equation}
    If $\hat{\gamma}G_0\geq\gamma_0$, we then always have $\gamma\geq\gamma_0$, and in this case we have:
    \begin{equation}
       \mc{P}(\gamma<\gamma_0|G_0<|g_ug_b|^2)\leq \mc{P}(\hat{\gamma}G_0<\gamma_0). 
    \end{equation}
    \item In the second case, $|g_ug_b|^2\leq G_0$, and $\mc{P}(\gamma<\gamma_0|G_0\geq|g_ug_b|^2)\leq 1$ (which is trivially true).
\end{enumerate}
By definition, the first and second case are mutually exclusive, and occur with probability $1-\delta$  and $\delta$, respectively. By applying the law of total probability, the overall probability $\mc{P}(\gamma<\gamma_0)$ is upper bounded by $\delta+(1-\delta)\mc{P}(G_0\hat{\gamma}<\gamma_0)$. As $\delta\leq1$, the theorem is proven.
\end{proof}

We can then consider power control, applying the bound in the first theorem to find the power requirement.

\begin{theorem}\label{th:power_control}
Let $\mc{G}(A_0):\{\mathbf{x}\in\mathbb{R}^2:|A(\boldsymbol{\phi}|\hat{\mathbf{x}}_u,\mathbf{x})|^2\geq A_0\}$ be the set of points for which the \gls{af} gain is larger than $A_0$, and let $\varepsilon=\mc{P}(\mb{x}_u\in\mc{G}(A_0)|\hat{\mb{x}}_u,\bm{\Sigma}_u)$ be the probability that the \gls{ue} is inside the set. As above, let $G_0>0$ be a positive value so that $\mc{P}(|g_ug_b|^2\leq G_0)=\delta$. We then have:
\begin{equation}
    \mc{P}(\gamma<\gamma_0)\leq\delta+\varepsilon,\,\, \forall P\geq\frac{\sigma^2\gamma_0}{N^2G_0A_0\min_{\mb{x}\in\mc{G}(A_0)}\beta(\mb{x})}.
\end{equation}
\end{theorem}
\begin{proof}
We know that the average \gls{snr} in a certain position is given by~\eqref{eq:av_snr}. If we consider a point in $\mc{G}(A_0)$, the average \gls{snr} is lower bounded by:
\begin{equation}
    \hat{\gamma}(\hat{\mb{x}}_u,\mb{x}_u)\geq \frac{P}{\sigma^2}N^2A_0\min_{\mb{x}\in\mc{G}(A_0)}\beta(\mb{x}),\quad \forall \mb{x}_u\in\mc{G}(A_0).
\end{equation}
Since the result above is a lower bound, we can divide the two cases:
\begin{enumerate}
    \item If $\mb{x}_u\in\mc{G}(A_0)$, we can set a value $P$ that ensures $G_0\hat{\gamma}>\gamma_0$ by applying the lower bound:
    \begin{equation}
        P\geq\frac{\sigma^2\gamma_0}{N^2G_0A_0\min_{\mb{x}\in\mc{G}(A_0)}\beta(\mb{x})}.
    \end{equation}
     We then have $\mc{P}(G_0\hat{\gamma}<\gamma_0|\mb{x}_u\in\mc{G}(A_0))=0$ for all power levels that satisfy the condition.
    \item If $\mb{x}_u\notin\mc{G}(A_0)$, we consider the packet as lost, i.e., we use the trivial bound $\mc{P}(G_0\hat{\gamma}<\gamma_0|\mb{x}_u\notin\mc{G}(A_0))\leq1$. By definition, this case occurs with probability $\varepsilon$.
\end{enumerate}
We then apply Theorem~\ref{th:prob_bound} and the law of total probability:
\begin{equation*}
\begin{aligned}
    \mc{P}(\gamma<\gamma_0)\leq&\delta+\mc{P}\left(\hat{\gamma}\leq\frac{\gamma_0}{G_0}\right)\\
    =&\delta+\varepsilon\mc{P}\left(\hat{\gamma}\leq\frac{\gamma_0}{G_0}\Big|\mb{x}_u\notin\mc{G}(A_0)\right)\\&+(1-\varepsilon)\mc{P}\left(\hat{\gamma}\leq\frac{\gamma_0}{G_0}\Big|\mb{x}_u\notin\mc{G}(A_0)\right)
    \leq\delta+\varepsilon,
\end{aligned}
\end{equation*}
completing the proof.
\end{proof}

\subsection{Beam Projection}

We now need to find a closed-form expression for region $\mc{G}(A_0)$. The region in 3D space that has an \gls{af} gain larger than $A_0$ is an elliptic cone, whose axis is the line between the \gls{ris} and $\hat{\mb{x}}_u$ and whose base is an ellipse whose axes are defined by the two beamwidth parameters from~\eqref{eq:3Dbeamwidth}. The projection of this elliptic cone on the plane defined by the floor (i.e., $z=h$) defines $\mc{G}(A_0)$. 

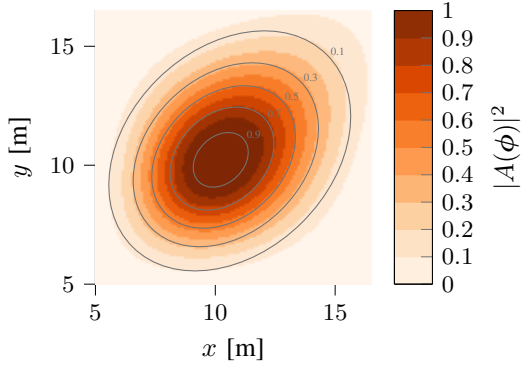
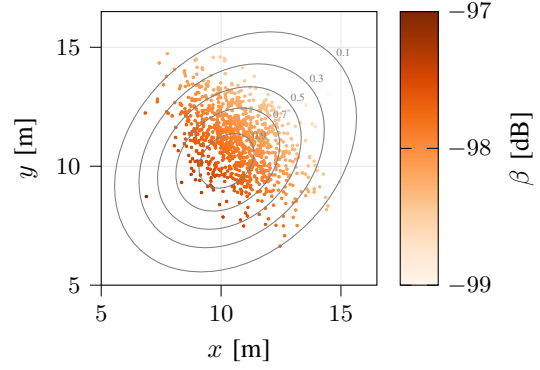
\begin{figure*}[hbt]
\centering
\subfloat[Heatmap of the \gls{af} gain.\label{fig:projection:heatmap}]{
\begin{tikzpicture}

    
\definecolor{forestgreen}{rgb}{0.13, 0.55, 0.13}

\definecolor{gray}{rgb}{0.450, 0.450, 0.45}

\begin{axis}[
width=\mwidth,
height=\mwidth,
scale only axis,
tick align=outside,
tick pos=left,
xlabel={\(\displaystyle x\) [m]},
xmin=5, xmax=16.5,
ylabel={\(\displaystyle y\) [m]},
ymin=5, ymax=16.5,
colormap/Oranges,
colorbar sampled,
colorbar style={
    samples=11, 
    ytick={0.0, 0.1, 0.2, 0.3, 0.4, 0.5, 0.6, 0.7, 0.8, 0.9, 1.0},
    ylabel={$|A(\bm{\phi})|^2$}, 
    ytick pos=right, 
    ytick align=inside,
    ytick style={color=gray},
    },
point meta min=0.0,
point meta max=1.0,
]
\addplot graphics [includegraphics cmd=\pgfimage,xmin=5, xmax=16.5, ymin=5, ymax=16.5] {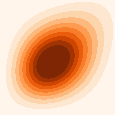};

\draw[draw=gray,rotate around={45:(axis cs:10.6033536581688,10.6033536581688)}] (axis cs:10.6033536581688,10.6033536581688) ellipse (5.72209035349177 and 4.26064583985972);

\draw[draw=gray,rotate around={45:(axis cs:10.4396702769788,10.4396702769788)}] (axis cs:10.4396702769788,10.4396702769788) ellipse (4.3784912175736 and 3.26988763549885);

\draw[draw=gray,rotate around={45:(axis cs:10.346620309982,10.346620309982)}] (axis cs:10.346620309982,10.346620309982) ellipse (3.39177078246033 and 2.53729531815614);

\draw[draw=gray,rotate around={45:(axis cs:10.2796782542677,10.2796782542677)}] (axis cs:10.2796782542677,10.2796782542677) ellipse (2.4514643388926 and 1.8361223974919);

\draw[draw=gray,rotate around={45:(axis cs:10.2268559702142,10.2268559702142)}] (axis cs:10.2268559702142,10.2268559702142) ellipse (1.29877355109335 and 0.9737108975288);

\draw (axis cs:14.649482549685,14.649482549685) node[
  scale=0.4165,
  anchor=base west,
  text=gray,
  rotate=0.0
]{0.1};
\draw (axis cs:13.5357311082908,13.5357311082908) node[
  scale=0.4165,
  anchor=base west,
  text=gray,
  rotate=0.0
]{0.3};
\draw (axis cs:12.7449644304901,12.7449644304901) node[
  scale=0.4165,
  anchor=base west,
  text=gray,
  rotate=0.0
]{0.5};
\draw (axis cs:12.0131253121357,12.0131253121357) node[
  scale=0.4165,
  anchor=base west,
  text=gray,
  rotate=0.0
]{0.7};
\draw (axis cs:11.145227555418,11.145227555418) node[
  scale=0.4165,
  anchor=base west,
  text=gray,
  rotate=0.0
]{0.9};

\end{axis}
\end{tikzpicture}}%
\qquad
\subfloat[Scatter plot of \gls{ue} positions $\mb{x}_u$.\label{fig:projection:scatter}]{
\begin{tikzpicture}

\definecolor{darkgray176}{RGB}{176,176,176}
\definecolor{darkgray}{rgb}{0.450, 0.450, 0.45}

\begin{axis}[
width=\mwidth,
height=\mwidth,
scale only axis,
tick align=outside,
tick pos=left,
xlabel={\(\displaystyle x\) [m]},
xmin=5, xmax=16.5,
ylabel={\(\displaystyle y\) [m]},
ymin=5, ymax=16.5,
xmajorgrids,
ymajorgrids,
colorbar,
colorbar style={ytick={0,0.5,1},
yticklabels={
  \(\displaystyle {-99}\),
  \(\displaystyle {-98}\),
  \(\displaystyle {-97}\)
},ylabel={$\beta$ [dB]}},
colormap/Oranges,
point meta min=0.0,
point meta max=1.0,
]
\addplot graphics [includegraphics cmd=\pgfimage,xmin=5, xmax=16.5, ymin=5, ymax=16.5] {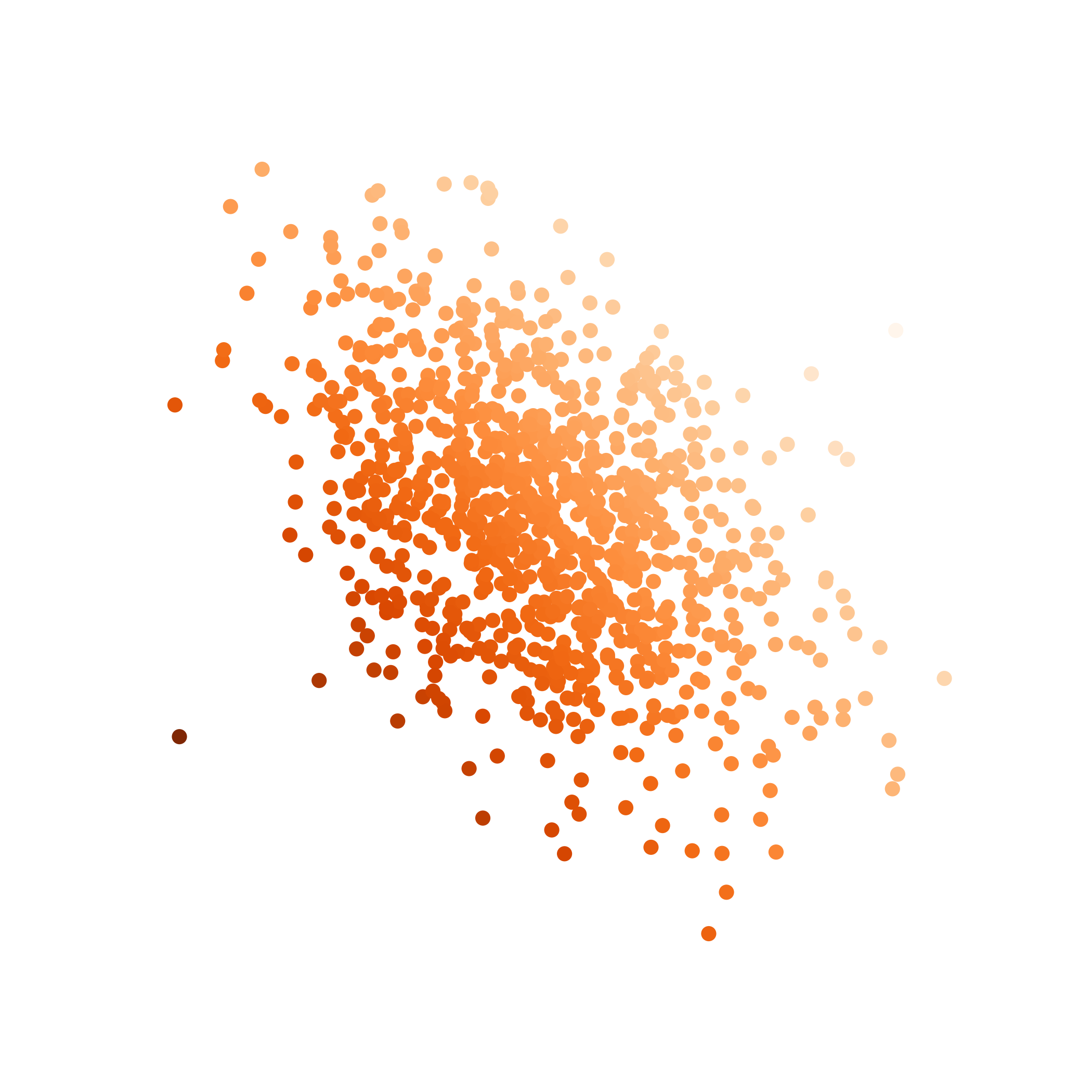};

\draw[draw=gray,rotate around={45:(axis cs:10.6033536581688,10.6033536581688)}] (axis cs:10.6033536581688,10.6033536581688) ellipse (5.72209035349177 and 4.26064583985972);

\draw[draw=gray,rotate around={45:(axis cs:10.4396702769788,10.4396702769788)}] (axis cs:10.4396702769788,10.4396702769788) ellipse (4.3784912175736 and 3.26988763549885);

\draw[draw=gray,rotate around={45:(axis cs:10.346620309982,10.346620309982)}] (axis cs:10.346620309982,10.346620309982) ellipse (3.39177078246033 and 2.53729531815614);

\draw[draw=gray,rotate around={45:(axis cs:10.2796782542677,10.2796782542677)}] (axis cs:10.2796782542677,10.2796782542677) ellipse (2.4514643388926 and 1.8361223974919);

\draw[draw=gray,rotate around={45:(axis cs:10.2268559702142,10.2268559702142)}] (axis cs:10.2268559702142,10.2268559702142) ellipse (1.29877355109335 and 0.9737108975288);

\draw (axis cs:14.649482549685,14.649482549685) node[
  scale=0.4165,
  anchor=base west,
  text=gray,
  rotate=0.0
]{0.1};
\draw (axis cs:13.5357311082908,13.5357311082908) node[
  scale=0.4165,
  anchor=base west,
  text=gray,
  rotate=0.0
]{0.3};
\draw (axis cs:12.7449644304901,12.7449644304901) node[
  scale=0.4165,
  anchor=base west,
  text=gray,
  rotate=0.0
]{0.5};
\draw (axis cs:12.0131253121357,12.0131253121357) node[
  scale=0.4165,
  anchor=base west,
  text=gray,
  rotate=0.0
]{0.7};
\draw (axis cs:11.145227555418,11.145227555418) node[
  scale=0.4165,
  anchor=base west,
  text=gray,
  rotate=0.0
]{0.9};
\end{axis}

\end{tikzpicture}}%
\caption{$\mc{G}(A_0)$ for different $A_0$ using the ellipse approximation with $h=25$ m, $\hat{\varphi} = \pi / 4$, $\hat{\theta}=\pi / 6$.}\vspace{-0.4cm}
\label{fig:projection}
\end{figure*}

We can write the equation of the cone in a standard form employing a new coordinate system $O_2$. $O_2$ is such that the $w$ axis is the cone's vertical axis, and $u$ and $v$ axes are parallel to the major and minor axis of the ellipse, respectively.
The rotation matrix of the transformation from $O_1$ to $O_2$ is
\begin{equation} \label{eq:rotation-matrix}
    \mb{R} = \begin{pmatrix}
    -\cos\hat{\theta}\cos\hat{\varphi} & -\cos\hat{\theta}\sin\hat{\varphi} & \sin\hat{\theta} \cos\hat{\varphi} \\
    \sin\hat{\varphi} & -\cos\hat{\varphi} & 0 \\
    \sin\hat{\theta}\cos\hat{\varphi} & \sin\hat{\theta}\sin\hat{\varphi} & \cos\hat{\theta} \\
    \end{pmatrix}.
\end{equation}
Hence, the point on the floor ($z = h$) of in $O_1$ can be represented in $O_2$ following the relation
\begin{equation} \label{eq:transform}
    \begin{pmatrix}
    u & v & w 
    \end{pmatrix}^T
    = \mb{R} 
    \begin{pmatrix}
    x & y & z   
    \end{pmatrix}^T.
\end{equation}
Using coordinate system $O_2$, the equation for the cone is
\begin{equation} \label{eq:cone-O2}
    \frac{u^2}{a^2}+\frac{v^2}{b^2}=w^2,
\end{equation}
where $a$ and $b$ can be computed from the definition of the 3D beamwidth angles given in~\eqref{eq:3Dbeamwidth} at reference distance $w=1$:
\begin{equation}
\begin{aligned}
    a =  \tan\left( \Delta\theta(A_0) /2 \right), \quad
    b = \tan\left( \Delta\varphi(A_0) / 2 \right).
\end{aligned}
\end{equation}
Therefore, we can substitute~\eqref{eq:transform} into~\eqref{eq:cone-O2} and compute the equation of the intersection of the cone with the floor, i.e.,
\begin{equation} \label{eq:proj_ellipse}
\begin{cases}
        Ax^2+Bxy+Cy^2+Dx+Ey+F=0;\\
        z=h.
\end{cases}
\end{equation}
where the parameters indicated by capital letters are given by:
\begin{equation}
\begin{cases}
    A&=\cos^2\hat{\varphi}\left(a^{-2}\cos^2\hat{\theta}- \sin^2\hat{\theta} \right) + b^{-2}\sin^2\hat{\varphi} \\
    B&=2\cos\hat{\varphi}\sin\hat{\varphi}\left(a^{-2}\cos^2\hat{\theta} - b^{-2} - \sin^2\hat{\theta} \right)\\
    C&= \sin^2\hat{\varphi}\left(a^{-2}\cos^2\hat{\theta}- \sin^2\hat{\theta}\right) + b^{-2}\cos^2\hat{\varphi}  \\
    D&=- 2h \cos\hat{\theta}\sin\hat{\theta}\cos\hat{\varphi} \left(a^{-2}\cos\hat{\varphi} + 1 \right)\\
    E&= -2h \cos\hat{\theta}\sin\hat{\theta} \sin\hat{\varphi} \left( a^{-2}\cos\hat{\varphi} + 1 \right)\\
    F&=h^2\left(a^{-2} \sin^2\hat{\varphi} \cos^2\hat{\varphi} + \cos^2\hat{\theta} \right).
\end{cases}
\end{equation}

The center and semi-axes of the projected ellipse are denoted as $\mb{x}_c$, $a'$ and $b'$, respectively, and can be easily derived from~\eqref{eq:proj_ellipse}. The ellipse's major semi-axis $a'$ is rotated by an angle $\hat{\varphi}$ from the $x$ axis. Fig.~\ref{fig:projection:heatmap} shows the heatmap of the \gls{af} gain and the projected ellipses obtained with different values of $A_0$: the heatmap shows the real \gls{af} gain, while the ellipses drawn on it represent the approximation given above. The approximation is generally good, if slightly pessimistic, guaranteeing that points inside the ellipse will respect the condition on the \gls{af} gain.

\subsection{Iterative Optimization}

By using beam projection, we can then compute the region $\mc{G}(A_0)$, which corresponds to an ellipse, and find the minimum power required in that region. If our goal is to find the minimum power $P$ that ensures a reliability level $p_s$, 
we can set $\delta$ and $\varepsilon$ such that $\delta + \varepsilon = 1 - p_s$ and run Algorithm~\ref{alg:power}.

\pgfplotsset{
width=3.5cm,
height=3.5cm,
scale only axis}
\begin{figure*}[bt]
\centering
    \subfloat[$\Psi = 0$ \label{fig:heatmap:0}]{
\begin{tikzpicture}

\begin{axis}[
tick align=outside,
tick pos=left,
xlabel={\(\displaystyle x\) [m]},
xmin=0, xmax=15,
xtick={0,2,4,6,8,10,12,14,16},
xticklabels={
  \(\displaystyle {0}\),
  \(\displaystyle {2}\),
  \(\displaystyle {4}\),
  \(\displaystyle {6}\),
  \(\displaystyle {8}\),
  \(\displaystyle {10}\),
  \(\displaystyle {12}\),
  \(\displaystyle {14}\),
  \(\displaystyle {16}\)
},
ylabel={\(\displaystyle y\) [m]},
ymin=0, ymax=15,
ytick={0,2,4,6,8,10,12,14,16},
yticklabels={
  \(\displaystyle {0}\),
  \(\displaystyle {2}\),
  \(\displaystyle {4}\),
  \(\displaystyle {6}\),
  \(\displaystyle {8}\),
  \(\displaystyle {10}\),
  \(\displaystyle {12}\),
  \(\displaystyle {14}\),
  \(\displaystyle {16}\)
}
]
\addplot graphics [includegraphics cmd=\pgfimage, xmin=0, xmax=15, ymin=0, ymax=15] {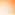};

\end{axis}

\end{tikzpicture}} \quad
    \subfloat[$\Psi = \pi/4$ \label{fig:heatmap:45}]{
\begin{tikzpicture}

\begin{axis}[
tick align=outside,
tick pos=left,
xlabel={\(\displaystyle x\) [m]},
xmin=0, xmax=15,
xtick style={color=black},
xtick={0,2,4,6,8,10,12,14,16},
xticklabels={
  \(\displaystyle {0}\),
  \(\displaystyle {2}\),
  \(\displaystyle {4}\),
  \(\displaystyle {6}\),
  \(\displaystyle {8}\),
  \(\displaystyle {10}\),
  \(\displaystyle {12}\),
  \(\displaystyle {14}\),
  \(\displaystyle {16}\)
},
ymin=0, ymax=15,
ytick style={color=black},
ytick={0,2,4,6,8,10,12,14,16},
yticklabels={
  \(\displaystyle {0}\),
  \(\displaystyle {2}\),
  \(\displaystyle {4}\),
  \(\displaystyle {6}\),
  \(\displaystyle {8}\),
  \(\displaystyle {10}\),
  \(\displaystyle {12}\),
  \(\displaystyle {14}\),
  \(\displaystyle {16}\)
}
]
\addplot graphics [includegraphics cmd=\pgfimage,xmin=0, xmax=15, ymin=0, ymax=15] {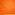};
\end{axis}

\end{tikzpicture}} \quad
    \subfloat[$\Psi = -\pi/4$\label{fig:heatmap:-45}]{
\begin{tikzpicture}

\begin{axis}[
colorbar,
colorbar style={
ytick pos=right,
ytick style={color=gray},
ytick={12.5,13,13.5,14,14.5,15,15.5,16,16.5},
yticklabels={
  \(\displaystyle {12.5}\),
  \(\displaystyle {13.0}\),
  \(\displaystyle {13.5}\),
  \(\displaystyle {14.0}\),
  \(\displaystyle {14.5}\),
  \(\displaystyle {15.0}\),
  \(\displaystyle {15.5}\),
  \(\displaystyle {16.0}\),
  \(\displaystyle {16.5}\)
},
ylabel={$P$ [dBm]}},
colormap/Oranges,
point meta max=16.4901176218056,
point meta min=12.7935849567194,
tick align=outside,
tick pos=left,
xlabel={\(\displaystyle x\) [m]},
xmin=0, xmax=15,
xtick={0,2,4,6,8,10,12,14,16},
xticklabels={
  \(\displaystyle {0}\),
  \(\displaystyle {2}\),
  \(\displaystyle {4}\),
  \(\displaystyle {6}\),
  \(\displaystyle {8}\),
  \(\displaystyle {10}\),
  \(\displaystyle {12}\),
  \(\displaystyle {14}\),
  \(\displaystyle {16}\)
},
ymin=0, ymax=15,
ytick={0,2,4,6,8,10,12,14,16},
yticklabels={
  \(\displaystyle {0}\),
  \(\displaystyle {2}\),
  \(\displaystyle {4}\),
  \(\displaystyle {6}\),
  \(\displaystyle {8}\),
  \(\displaystyle {10}\),
  \(\displaystyle {12}\),
  \(\displaystyle {14}\),
  \(\displaystyle {16}\)
}
]
\addplot graphics [includegraphics cmd=\pgfimage,xmin=0, xmax=15, ymin=0, ymax=15] {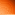};
\end{axis}

\end{tikzpicture}}
    \caption{Power consumption in the first quadrant of the area.}\vspace{-0.4cm}
    \label{fig:heatmap}
\end{figure*}

\setlength{\textfloatsep}{0pt}
\begin{algorithm}[t]
\caption{Power control optimization.}
\label{alg:power}
\footnotesize
 \SetKwFunction{Fn}{PowerControl}
  \Fn{$B,L,T,N,\sigma,\hat{\mb{x}}_u, \bm{\Sigma}_u, A_{\min}, \delta, \varepsilon, \nu$}
  {
  \\
  $\gamma_0\leftarrow 2^{\frac{L}{BT}}-1$\;\label{line:gamma0}
  $G_0\leftarrow$\Call{FadingICDF}{$\delta$}\;\label{line:fading}
  $A_0\leftarrow$\Call{ArrayFactor}{$\hat{\mb{x}}_u, \bm{\Sigma}_u, A_{\min}, \varepsilon, \nu$}\;
  \If{$A_0>0$}{
  $\mb{x}_c,a',b'\leftarrow$\Call{EllipseParameters}{$\hat{\mb{x}}_u,A_0$}\;
  $\hat{\beta}\leftarrow$\Call{PathLoss}{$\mb{x}_c + a' (\cos\hat{\varphi}, \sin\hat{\varphi}, 0)\T$}\;\label{line:loss}
  $P\leftarrow\frac{\sigma^2\gamma_0}{N^2 G_0  A_0 \hat{\beta}}$\;
   \Return{$A_{\ell}$}\;
   }
   \Else{
    \Return{$-1$};
   }
}
\end{algorithm}

First, the value of $\gamma_0$ is computed by applying~\eqref{eq:min_snr}, while $G_0$ is computed in line~\ref{line:fading}: as the overall fading distribution is determined by the product of two independent Rice fading gains, $|g_u g_b|^2$, the inverse \gls{cdf} cannot be computed analytically. However, fading parameters are relatively stable, and can be computed in advance by numerical methods and tabulated, so that the calculation just requires the retrieval of the correct value from a table. Finding the value of $A_0$ that ensures that $\mc{P}(\mb{x}_u\notin\mc{G}(A_0))<\varepsilon$ requires an iterative optimization, which can be performed by binary search. The precision parameter $\nu$ determines the number of iterations that the search will use, but since $A_0$ is directly proportional to the required power, we do not need a large number of iterations. Considering a minimum \gls{af} gain $A_{\min}=0.1$, which ensures that the illuminated points are still within the main lobe, just 5 iterations allow us to reach a maximum error of 2\% on the value of the array factor. If the positioning uncertainty is too large, and even $\mc{G}(A_{\min})$ is too small for $\varepsilon$, the \gls{urllc} transmission is impossible, and the \gls{bs} reports this to the application. Taking a pessimistic approach, we consider the minimum viable array factor within this precision, giving us a worst case increase of transmit power of 2\%. The full binary search procedure is reported in Algorithm~\ref{alg:gain_search}: the {\textsc{EllipseProbability}} function, first used in line~\ref{line:ellipse_prob} of the algorithm, simply computes the probability that the \gls{ue} will be inside the ellipse by computing the \gls{cdf} of the position distribution (see Fig.~\ref{fig:projection:scatter} for a visualization of the position realizations and $\mc{G}(A_0)$). While the integral of a bivariate Gaussian random variable over an arbitrary ellipse --which is required to get $\mc{P}(\mb{x}_u\in\mc{G}(A_0))$-- does not have a closed-form solution, it is a well-known numerical problems with several efficient and tabulated  solutions~\cite{didonato1961integration,groenewoud1967bivariate}.

\setlength{\textfloatsep}{0pt}
\begin{algorithm}[t]
\caption{Reliable AF gain binary search.}
\label{alg:gain_search}
\footnotesize
 \SetKwFunction{Fn}{ArrayFactor}
  \Fn{$\hat{\mb{x}}_u, \bm{\Sigma}_u, A_{\min}, N, \varepsilon, \nu$}{
  \\
  $A_h\leftarrow 1-\nu$\;
  $A_\ell\leftarrow A_{\min}$\;
  \If {\Call{EllipseProbability}{$\hat{\mb{x}}_u, \bm{\Sigma}_u, A_{\ell},N$}$<1-\varepsilon$}{\label{line:ellipse_prob}
   \Return{$-1$}\;
  }
  \If {\Call{EllipseProbability}{$\hat{\mb{x}}_u, \bm{\Sigma}_u, A_h,N$}$\geq1-\varepsilon$}{
   \Return{$1-\nu$}\;
  }
   \tcc{Binary search}
   \While{$A_h-A_{\ell}>\nu$}{
      $A_0\leftarrow(A_h+A_{\ell})/2$\;
      $e\leftarrow$\Call{EllipseProbability}{$\hat{\mb{x}}_u, \bm{\Sigma}_u, A_0,N$}\;
      \If{$e<1-\varepsilon$}{
         $A_h\leftarrow A_0$\;
      }
      \Else{
        $A_{\ell}\leftarrow A_0$\;
      }
   }
   \Return{$A_{\ell}$}\;
}
\end{algorithm}

Finally, since $\mc{G}(A_0)$ is approximated by the ellipse computed in the previous section, the maximum attenuation given by the path loss in the region is 
\begin{equation}
\min_{\mb{x}\in\mc{G}(A_0)} \beta(\mb{x}) = \beta(\mb{x}_c + a' (\cos\hat{\varphi}, \sin\hat{\varphi}, 0)\T).
\end{equation}

As all numerical steps in the optimization can be tabulated and computed efficiently, and the number of iterations of the binary search is extremely limited; the procedure for power control based on the upper bound can be computed within the \gls{urllc} time constraints, even if we consider the limits of embedded processors that can be installed in a \gls{bs}.

\section{Numerical Results}\label{sec:results}
In this section, we present our numerical results. 
The parameters we used are listed in Table~\ref{tab:params}, and may refer to a typical \gls{urllc} scenario. The values of $\delta$ and $\varepsilon$ have been found empirically. We highlight that the deadline to transmit the packet is $0.5$~ms, which is stricter than the typical \gls{urllc} deadline to allow some time for \gls{ris} configuration and power optimization: however, this is still not enough to perform the extensive Monte Carlo simulations that would be required to compute the optimal power. 
The covariance matrix of the \gls{ue}'s position uncertainty is
\begin{equation} \label{eq:Sigma}
    \mb{\Sigma} = \sigma^2_u 
    \begin{bmatrix}
    \frac{1}{\cos^2\Psi} & \sin\Psi \\
    \sin\Psi & \frac{1}{\cos^2\Psi} \\
    \end{bmatrix}
\end{equation}
where $\sigma_u = 0.3$ m and $\Psi \in \{0, \pi/4, -\pi/4\}$. Eq.~\eqref{eq:Sigma} is used to capture three different user behaviors: when $\Psi = 0$, the uncertainty is a circularly symmetric Gaussian representing the error when the user is static; when $\Psi = \pi/4$ or $-\pi/4$, the major axis of the equi-probability ellipse is oriented toward $\Psi$, emulating the output error of common tracking filters when the \gls{ue} is moving in direction $\Psi$~\cite{li2020toward}. The worst-case scenario is when $\Psi=\pi/2-\hat{\varphi}$, as the highest positioning error is aligned with the minor axis of the projected beam ellipse.

\setlength{\textfloatsep}{15pt}
\begin{table}[t]
    \centering
    \caption{Simulation parameters.}\vspace{-0.2cm}
    \footnotesize
    {\renewcommand{\arraystretch}{1}
    \begin{tabular}{@{}lcc@{}}
    \toprule
    Parameter & Symbol & Value \\ \midrule
    \multicolumn{3}{c}{\textbf{Scenario}} \\ \midrule
    Room side & $D$ & $15$~m\\ 
    Ceiling height &$h$ & $25$~m\\
    BS position & $\mb{x}_b$ & $(-5, -5, 5)\T$~m\\
    RIS element spacing & $d$ & $\lambda / 2$\\
    Number of RIS elements & $N$ & 100\\
    Positioning error deviation & $\sigma_u$ & $0.3$~m\\ 
    \midrule
     \multicolumn{3}{c}{\textbf{Communication system}}\\ \midrule
     Packet length & $L$ & $32$~bytes\\
     Wavelength & $\lambda$ & 0.333 m\\
     Bandwidth & $B$ & $360$~kHz\\
     Latency constraint & $T$ & $0.5$~ms\\
     Reliability & $p_s$ & $99.999$\% \\
     Fading shape parameter & $K$ & $6$~dB\\
     \gls{ue} and \gls{bs} antenna gains &$G_b \cdot G_u$ & $12.85$~dB\\
     Reference distance &$d_0$ & $1$~m\\
     Path loss exponent &$\xi$ &$2$\\
     Reference path gain & $\beta_0$ & $-31.53$~dB\\
     Algorithm parameters & $[\delta, \varepsilon]$ & $[0.9, 0.1] \cdot (1 - p_s)$ \\
     \bottomrule
    \end{tabular}\vspace{-0.4cm}}
    \label{tab:params}
\end{table}

Fig.~\ref{fig:heatmap} shows a heatmap of power consumption for the three values of $\Psi$: we can easily notice that power consumption is generally lower for $\Psi=0$, and that the path loss is still the most important component: points farther away from the origin generally require a higher power. However, the increase is slower than the path loss, as the effect of the positioning error decreases: as $\hat{\theta}$ increases, the projection of the beam on the floor becomes larger, so the same positioning error distribution is covered by an ellipse associated to a higher \gls{af} gain. In the asymmetric error cases, shown in Fig.~\ref{fig:heatmap:45}-\subref*{fig:heatmap:-45}, a the effect of $\hat{\varphi}$ is significant: if $\Psi=\pi/4$, the required power is minimal when $\hat{\varphi}=\Psi$, i.e., when the projected beam and the positioning error are aligned, and increases as the two ellipses rotate relative to each other. The opposite happens if $\Psi=-\pi/4$, as required power is maximal when the highest position error is orthogonal to the projected beam. Interestingly, the case with $\Psi=\pi/4$ also requires less power when the \gls{ue} is farther away: if $\varphi$ is close to $\Psi$, the increased eccentricity of the projection beam better matches the shape of the position distribution, improving the reliable \gls{af} gain $A_0$ enough to offset the increased path loss. Naturally, the opposite happens if $\Psi=-\pi/4$.

\pgfplotsset{
width=6.3cm,
height=3cm,
scale only axis}
\begin{figure*}[bt]
    \centering
    \subfloat[Transmission power. \label{fig:power}]{\input{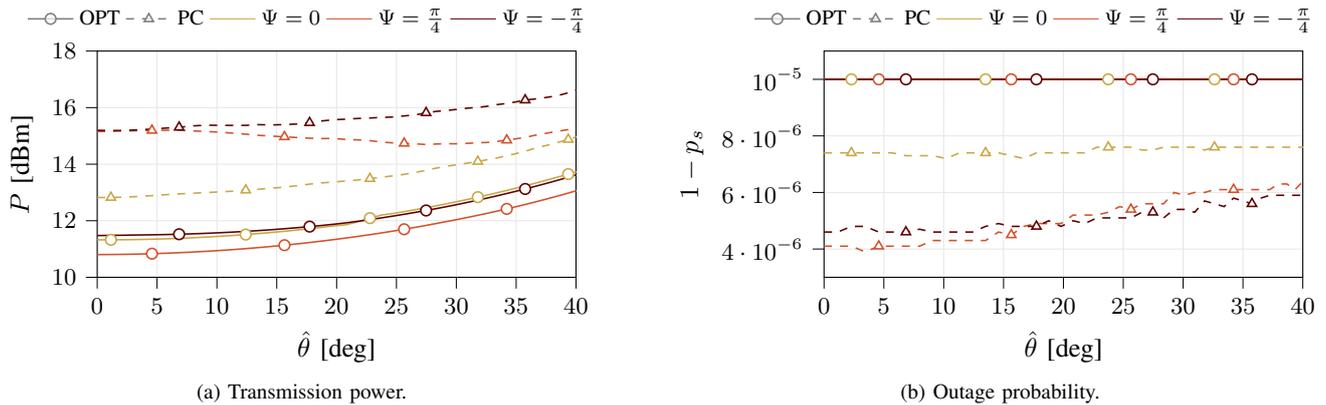}}$\qquad$
    \subfloat[Outage probability. \label{fig:outage}]{
\begin{tikzpicture}

\begin{axis}[
scale only axis,
legend cell align={left},
legend style={
  at={(1.05,1.05)},
  anchor=south east,
  draw=white,
},
legend columns=-1,
log basis y={10},
tick align=outside,
tick pos=left,
xlabel={\(\displaystyle \hat{\theta}\) [deg]},
xmajorgrids,
xmin=0, xmax=40,
xtick style={color=black},
ylabel={\(\displaystyle 1 - p_s\)},
ymajorgrids,
ymin=3e-06, ymax=1.1e-05,
ytick scale label code/.code={},
ytick={2e-06,4e-06,6e-6,8e-6,1e-05},
yticklabels={
  \(\displaystyle {2\cdot10^{-6}}\),
  \(\displaystyle {4\cdot10^{-6}}\),
  \(\displaystyle {6\cdot10^{-6}}\),
  \(\displaystyle {8\cdot 10^{-6}}\),
  \(\displaystyle {10^{-5}}\),
},
xtick={-5,0,5,10,15,20,25,30,35,40,45},
xticklabels={
  \(\displaystyle {\ensuremath{-}5}\),
  \(\displaystyle {0}\),
  \(\displaystyle {5}\),
  \(\displaystyle {10}\),
  \(\displaystyle {15}\),
  \(\displaystyle {20}\),
  \(\displaystyle {25}\),
  \(\displaystyle {30}\),
  \(\displaystyle {35}\),
  \(\displaystyle {40}\),
  \(\displaystyle {45}\)
},
]
\addplot [semithick, color1, mark=*, mark size=2, mark options={solid, fill=white}, mark repeat=10, mark phase=3, forget plot]
table {%
0 1e-05
1.14576283817483 1e-05
2.29061004263859 1e-05
3.43363036245061 1e-05
4.5739212599009 1e-05
5.71059313749963 1e-05
6.84277341263097 1e-05
7.96961039432135 1e-05
9.09027692082235 1e-05
10.2039737217317 1e-05
11.3099324740202 1e-05
12.4074185274007 1e-05
13.4957332807958 1e-05
14.5742161980387 1e-05
15.6422464572087 1e-05
16.6992442339936 1e-05
17.7446716250569 1e-05
18.7780332224455 1e-05
19.7988763545249 1e-05
20.8067910127112 1e-05
21.8014094863518 1e-05
22.7824057304817 1e-05
23.7494944928668 1e-05
24.7024302277713 1e-05
25.6410058243053 1e-05
26.565051177078 1e-05
27.4744316262771 1e-05
28.3690462932786 1e-05
29.248826336547 1e-05
30.1137331509824 1e-05
30.9637565320735 1e-05
31.7989128242944 1e-05
32.6192430711928 1e-05
33.4248111826038 1e-05
34.2157021324374 1e-05
34.9920201985587 1e-05
35.7538872544367 1e-05
36.5014411205063 1e-05
37.2348339815747 1e-05
37.9542308751325 1e-05
38.6598082540901 1e-05
39.3517526262647 1e-05
40.0302592718897 1e-05
};

\addplot [semithick, color1, dashed, mark=triangle*, mark size=2, mark options={solid, fill=white}, mark repeat=10, mark phase=3, forget plot]
table {%
0 7.4e-06
1.14576283817483 7.4e-06
2.29061004263859 7.4e-06
3.43363036245061 7.4e-06
4.5739212599009 7.4e-06
5.71059313749963 7.4e-06
6.84277341263097 7.3e-06
7.96961039432135 7.3e-06
9.09027692082235 7.3e-06
10.2039737217317 7.2e-06
11.3099324740202 7.4e-06
12.4074185274007 7.4e-06
13.4957332807958 7.4e-06
14.5742161980387 7.4e-06
15.6422464572087 7.3e-06
16.6992442339936 7.2e-06
17.7446716250569 7.4e-06
18.7780332224455 7.4e-06
19.7988763545249 7.4e-06
20.8067910127112 7.4e-06
21.8014094863518 7.4e-06
22.7824057304817 7.6e-06
23.7494944928668 7.6e-06
24.7024302277713 7.6e-06
25.6410058243053 7.6e-06
26.565051177078 7.5e-06
27.4744316262771 7.5e-06
28.3690462932786 7.6e-06
29.248826336547 7.6e-06
30.1137331509824 7.6e-06
30.9637565320735 7.6e-06
31.7989128242944 7.5e-06
32.6192430711928 7.6e-06
33.4248111826038 7.6e-06
34.2157021324374 7.6e-06
34.9920201985587 7.6e-06
35.7538872544367 7.6e-06
36.5014411205063 7.6e-06
37.2348339815747 7.6e-06
37.9542308751325 7.6e-06
38.6598082540901 7.6e-06
39.3517526262647 7.6e-06
40.0302592718897 7.6e-06
};
\addplot [semithick, color2, mark=*, mark size=2, mark options={solid, fill=white}, mark repeat=10, mark phase=5, forget plot]
table {%
0 1e-05
1.14576283817483 1e-05
2.29061004263859 1e-05
3.43363036245061 1e-05
4.5739212599009 1e-05
5.71059313749963 1e-05
6.84277341263097 1e-05
7.96961039432135 1e-05
9.09027692082235 1e-05
10.2039737217317 1e-05
11.3099324740202 1e-05
12.4074185274007 1e-05
13.4957332807958 1e-05
14.5742161980387 1e-05
15.6422464572087 1e-05
16.6992442339936 1e-05
17.7446716250569 1e-05
18.7780332224455 1e-05
19.7988763545249 1e-05
20.8067910127112 1e-05
21.8014094863518 1e-05
22.7824057304817 1e-05
23.7494944928668 1e-05
24.7024302277713 1e-05
25.6410058243053 1e-05
26.565051177078 1e-05
27.4744316262771 1e-05
28.3690462932786 1e-05
29.248826336547 1e-05
30.1137331509824 1e-05
30.9637565320735 1e-05
31.7989128242944 1e-05
32.6192430711928 1e-05
33.4248111826038 1e-05
34.2157021324374 1e-05
34.9920201985587 1e-05
35.7538872544367 1e-05
36.5014411205063 1e-05
37.2348339815747 1e-05
37.9542308751325 1e-05
38.6598082540901 1e-05
39.3517526262647 1e-05
40.0302592718897 1e-05
};
\addplot [semithick, color2, dashed, mark=triangle*, mark size=2, mark options={solid, fill=white}, mark repeat=10,mark phase=5, forget plot]
table {%
0 4.1e-06
1.14576283817483 4.1e-06
2.29061004263859 4.1e-06
3.43363036245061 3.9e-06
4.5739212599009 4.1e-06
5.71059313749963 4.1e-06
6.84277341263097 4.1e-06
7.96961039432135 4.1e-06
9.09027692082235 4.3e-06
10.2039737217317 4.3e-06
11.3099324740202 4.3e-06
12.4074185274007 4.3e-06
13.4957332807958 4.3e-06
14.5742161980387 4.6e-06
15.6422464572087 4.5e-06
16.6992442339936 4.8e-06
17.7446716250569 4.9e-06
18.7780332224455 4.9e-06
19.7988763545249 4.9e-06
20.8067910127112 5.2e-06
21.8014094863518 5.2e-06
22.7824057304817 5.2e-06
23.7494944928668 5.3e-06
24.7024302277713 5.5e-06
25.6410058243053 5.4e-06
26.565051177078 5.6e-06
27.4744316262771 5.6e-06
28.3690462932786 5.6e-06
29.248826336547 6e-06
30.1137331509824 5.9e-06
30.9637565320735 6e-06
31.7989128242944 6e-06
32.6192430711928 6.1e-06
33.4248111826038 6.1e-06
34.2157021324374 6.1e-06
34.9920201985587 6.1e-06
35.7538872544367 6.1e-06
36.5014411205063 6.1e-06
37.2348339815747 6.1e-06
37.9542308751325 6.3e-06
38.6598082540901 6.3e-06
39.3517526262647 6.1e-06
40.0302592718897 6.4e-06
};
\addplot [semithick, color3, mark=*, mark size=2, mark options={solid, fill=white}, mark repeat=10,mark phase=7, forget plot]
table {%
0 1e-05
1.14576283817483 1e-05
2.29061004263859 1e-05
3.43363036245061 1e-05
4.5739212599009 1e-05
5.71059313749963 1e-05
6.84277341263097 1e-05
7.96961039432135 1e-05
9.09027692082235 1e-05
10.2039737217317 1e-05
11.3099324740202 1e-05
12.4074185274007 1e-05
13.4957332807958 1e-05
14.5742161980387 1e-05
15.6422464572087 1e-05
16.6992442339936 1e-05
17.7446716250569 1e-05
18.7780332224455 1e-05
19.7988763545249 1e-05
20.8067910127112 1e-05
21.8014094863518 1e-05
22.7824057304817 1e-05
23.7494944928668 1e-05
24.7024302277713 1e-05
25.6410058243053 1e-05
26.565051177078 1e-05
27.4744316262771 1e-05
28.3690462932786 1e-05
29.248826336547 1e-05
30.1137331509824 1e-05
30.9637565320735 1e-05
31.7989128242944 1e-05
32.6192430711928 1e-05
33.4248111826038 1e-05
34.2157021324374 1e-05
34.9920201985587 1e-05
35.7538872544367 1e-05
36.5014411205063 1e-05
37.2348339815747 1e-05
37.9542308751325 1e-05
38.6598082540901 1e-05
39.3517526262647 1e-05
40.0302592718897 1e-05
};
\addplot [semithick, color3, dashed, mark=triangle*, mark size=2, mark options={solid, fill=white}, mark repeat=10, mark phase=7, forget plot]
table {%
0 4.6e-06
1.14576283817483 4.6e-06
2.29061004263859 4.8e-06
3.43363036245061 4.8e-06
4.5739212599009 4.6e-06
5.71059313749963 4.6e-06
6.84277341263097 4.6e-06
7.96961039432135 4.7e-06
9.09027692082235 4.6e-06
10.2039737217317 4.6e-06
11.3099324740202 4.6e-06
12.4074185274007 4.6e-06
13.4957332807958 4.8e-06
14.5742161980387 4.9e-06
15.6422464572087 4.8e-06
16.6992442339936 4.8e-06
17.7446716250569 4.8e-06
18.7780332224455 5e-06
19.7988763545249 4.8e-06
20.8067910127112 5e-06
21.8014094863518 4.9e-06
22.7824057304817 5.1e-06
23.7494944928668 5.1e-06
24.7024302277713 5.1e-06
25.6410058243053 5.1e-06
26.565051177078 5.3e-06
27.4744316262771 5.3e-06
28.3690462932786 5.1e-06
29.248826336547 5.4e-06
30.1137331509824 5.4e-06
30.9637565320735 5.3e-06
31.7989128242944 5.7e-06
32.6192430711928 5.6e-06
33.4248111826038 5.5e-06
34.2157021324374 5.8e-06
34.9920201985587 5.7e-06
35.7538872544367 5.6e-06
36.5014411205063 5.8e-06
37.2348339815747 5.9e-06
37.9542308751325 5.9e-06
38.6598082540901 5.9e-06
39.3517526262647 5.9e-06
40.0302592718897 5.9e-06
};

\addlegendimage{semithick, gray, mark=*, mark options={solid, fill=white}}
\addlegendentry{OPT}

\addlegendimage{semithick, gray, dashed, mark=triangle*, mark options={solid, fill=white}}
\addlegendentry{PC}

\addlegendimage{color1}
\addlegendentry{$\Psi = 0$}
\addlegendimage{color2}
\addlegendentry{$\Psi = \frac{\pi}{4}$}
\addlegendimage{color3}
\addlegendentry{$\Psi = -\frac{\pi}{4}$}

\end{axis}
\end{tikzpicture}}
    \caption{Performance as a function of $\hat{\theta}$ with $\hat{\varphi} = \pi/4$.}\vspace{-0.4cm}
    \label{fig:line_results}
\end{figure*}


In order to provide a meaningful comparison, we consider a fixed azimuth angle $\hat{\varphi} = \pi/4$ and consider performance as a function of $\hat{\theta}\in [0,40^{\circ}]$, evaluating both the outage probability and power consumption. As a term of comparison, we show the results of an ideal oracle approach obtained by Monte Carlo simulations, referred as OPT in the following. For this solution, the empirical \gls{cdf} of $\gamma / P$ is estimated over $10^7$ realizations of the \gls{ue}'s position and fading, and inverted to find the power that gives exactly $\mc{P}(\gamma \le \gamma_0) = 1 - p_s$. Fig.~\ref{fig:line_results} shows the comparison, where the proposed practical power control solution is denoted as PC.

Fig.~\ref{fig:power} shows that the proposed solution slightly overestimates the required power, particularly in the cases with $\Psi=\pm\pi/4$: the optimality gap is between 1.5 and 4.5~dB, as the upper bound to the outage probability is relatively tight, but still leaves some slack. Fig.~\ref{fig:outage} confirms the conservative nature of the proposed solution: the outage probability, which should be close to $10^{-5}$, is between 25 and 60\% lower. It is worth noting that OPT always has an outage probability exactly equal to $10^{-5}$, as it obtained by the straightforward inversion of the empirical \gls{cdf} of $\gamma / P$, while, as we remarked above, PC provides the closest performance when the \gls{ue} is static, ($\Psi = 0$), and the worst performance when the direction of movement is perpendicular to $\hat{\varphi}$ ($\Psi = -\pi/4$). As we noted above, the increased eccentricity of $\mc{G}(A_0)$ as the elevation angle grows has a positive effect if $\Psi=\hat{\varphi}$, and a negative effect in the other two cases. When $\hat{\theta}\simeq0$, $\mc{G}(A_0)$ is almost circular: asymmetric uncertainty distributions are likelier to generate positions outside the illuminated area. When pointing toward the sides, $\mc{G}(A_0)$ is eccentric in the $\hat{\varphi}$ direction, and an asymmetric positioning error in the same direction is beneficial (see also Fig.~\ref{fig:projection:scatter} for the $\Psi = -\pi/4$ case).

\section{Conclusion and Discussion}\label{sec:conc}

In this work, we present a method for computationally efficient power control in \gls{ris}-assisted \gls{urllc}, based on an upper bound to the outage probability which includes position uncertainty. The results show that the optimality gap is relatively small, as the method always meets \gls{urllc} requirements with a slightly higher power than the optimum.

Future work can focus on further refining the bounds, as well as considering different localization methods, which may even rely on the \gls{ris} itself. Furthermore, an analysis of the calibration process, which the \gls{bs} uses to learn the locations in which the controllable path is dominating the uncontrollable paths, could be another extension of the paper. The optimization of the sweeping process that gauges the impact of the controllable path in each position in a dynamic environment is an interesting development.

\bibliographystyle{IEEEtran}
\bibliography{positioning.bib}

\end{document}